\documentclass[review,onefignum,onetabnum]{msiapack}

\usepackage{boxedminipage}
\usepackage{indentfirst}
\usepackage{booktabs} 
\usepackage{array}
\usepackage{amssymb}
\usepackage{rotating}
\usepackage{enumitem}
\usepackage{hyperref}
\usepackage{subcaption}
\usepackage{multirow}
\usepackage{graphicx, balance}
\usepackage{amsfonts}
\usepackage{epstopdf}
\usepackage{algorithmic}
\usepackage{rotating}

\ifpdf
  \DeclareGraphicsExtensions{.eps,.pdf,.png,.jpg}
\else
  \DeclareGraphicsExtensions{.eps}
\fi

\newsiamremark{remark}{Remark}
\newsiamremark{hypothesis}{Hypothesis}
\crefname{hypothesis}{Hypothesis}{Hypotheses}
\newsiamthm{claim}{Claim}

\headers{System-wide Diversity}{A. Antikacioglu, T. Bajpai, R. Ravi}

\title{A new system-wide diversity measure for recommendations with efficient algorithms}

\author{Arda Antikacioglu
  \thanks{(\email{antikacioglu@gmail.com}).}
\and Tanvi Bajpai\thanks{School of Computer Science, Carnegie Mellon University, Pittsburgh, PA.
 (\email{tbaj@cmu.edu}).}
\and R. Ravi\thanks{Tepper School of Business, Carnegie Mellon University, Pittsburgh, PA. (\email{ravi@cmu.edu})}}

\ifpdf
\hypersetup{
  pdftitle={A new system-wide diversity measure for recommendations with efficient algorithms},
  pdfauthor={A. Antikacioglu, T. Bajpai, R. Ravi}
}
\fi

\begin{document}

\maketitle

\begin{abstract}
  Recommender systems often operate on item catalogs clustered by genres, and user bases that have natural clusterings into user types by demographic or psychographic attributes. Prior work on system-wide diversity has mainly focused on defining intent-aware metrics among such categories and maximizing relevance of the resulting recommendations, but has not combined the notions of diversity from the two point of views of items and users.
In this work, (1) we introduce two new system-wide diversity metrics to simultaneously address the problems of diversifying the categories of items that each user sees, diversifying the types of users that each item is shown, and maintaining high recommendation quality.
We model this as a subgraph selection problem on the bipartite graph of candidate recommendations between users and items. (2) In the case of disjoint item categories and user types, we show that the resulting problems can be solved exactly in polynomial time, by a reduction to a minimum cost flow problem. (3) In the case of non-disjoint categories and user types, we prove NP-completeness of the objective and present efficient approximation algorithms using the submodularity of the objective.  (4) Finally, we validate the effectiveness of our algorithms on the MovieLens-1m and Netflix datasets, and show that algorithms designed for our objective also perform well on sales diversity metrics, and even some intent-aware diversity metrics. Our experimental results justify the validity of our new composite diversity metrics.
\end{abstract}



\section{Introduction}

The goal in the design of traditional recommendation systems is the accuracy of predictions as measured by the implied relevance of the recommended items. Collaborative filtering recommender systems are prone to providing item recommendations that are clustered in a filter-bubble~\cite{pariser2011filter} due to a rich-get-richer effect of commonly seen and rated items~\cite{fleder2009blockbuster}.
One potential `unsupervised' approach to address this may be to require some sort of expansion properties on the bipartite graph between users and items that are recommended to them.
However, in prior work, 
the various methods that have been proposed to diversify such recommendations typically focus on more targeted approaches such as increasing item exposure, re-ranking CF recommendations for diversity, or choosing appropriate subgraphs that reflect diversity metrics.

Often, such recommendation systems operate on item catalogs and user bases that have natural clusterings into item categories and user types. This single-minded focus on relevance fails to incorporate requirements of diversity of the recommendations among the item categories and user types. 
In this paper, we build on earlier targeted approaches for increasing diversity and 
propose a new model to achieve a holistic trade-off between user and item level diversity that also promotes system-level diversity. Our problem is motivated by three different considerations in incorporating such diversity in the design of recommender systems.

First, there is need for diversity in user's recommendation lists both in terms of items and categories, both to encourage serendipity \cite{tulving1994novelty} as well as improve user satisfaction~\cite{mcnee2006being}.

The second consideration we look at is item-level diversity, that is, we wish to show each item to a diverse set of users. Item-level diversity allows for a more holistic dissemination of items to users. User-level diversity fails to consider item-level diversity, since assigning recommendations based on user-satisfaction would still only show items to users who fall in their traditional ``audience." This would result in bad item-level diversity but could still give a high user-level diversity if recommendations are diverse enough.

Finally, the third consideration we have is system level diversity, which involves aggregating the recommendations made to all users and studying the resulting distribution of recommendations. The platform running the recommender system often has concerns other than pure user satisfaction or item-level diversity. Examples of such concerns include achieving good coverage of different categories in the item catalog and avoiding the perpetuation of biases across the system such as popularity bias or filter bubbles among demographic or psychographic clusters of users.

Typically, all three of these considerations are studied under the same ``diversity'' umbrella. However, systems that optimize for user or item level diversity do not necessarily score well in system-level diversity, motivating the need for a new objective that combines all of these considerations.

One common problem with deliberately increasing the diversity of recommendation systems is that they can change user behavior, which then changes future estimates of relevance, and hence eventually lead to polarization. 
We do not address this meta-concern in any detail so our targeted approach will also suffer from this same problem.
Nevertheless, targeting holistic diversity is a first step in this direction since the ensuing changes will not be necessarily biased in terms of changing only user behavior or the clustering of values of item relevances, but some mixture of the two.

\section{Related Work and Contributions}
First we survey previous work on category-aware metrics for diversification,
and then survey work on sales diversity measures that are system-wide measures of diversification. \ifdefined\THESIS Finally we review some related work on submodularity and crowd-sourcing.\fi
\subsection{Category-Aware Metrics}
Previous work has used category information in defining metrics for measuring the diversity contained in user lists. In our work, we focus on three, each  of which informs one of the baseline algorithms we compare against in our experimental section. \\

\textbf{Intra-list Distance (ILD):} We define a recommendation set's intra-list distance as the average pairwise distance among items \cite{carbonell1998use}. This is used to measure the diversity of an individual user's recommendations and quantifies user-novelty. The distance $dist(v_k,v_j)$ between items we consider is measured using the cosine similarity between the items' category membership vectors. Given a list $L$ of recommendations, defined by item lists of length $c_u$ for user $u$,
the intra-list distance is defined as follows (we use $L$ to denote the left-side of the bipartite graph representation representing the users, and $N(u_i)$ to represent the neighbors of user $u_i$, which are items in the right-hand side recommended to her).

\[ ILD = \frac{1}{|L|}\sum_{u_i \in L} \frac{1}{c_i(c_i-1)}\sum_{(v_k,v_j) \in N(u_i)} dist(v_k,v_j).\]

Maximizing this objective enforces items in the recommendation list of a user to be dissimilar, but ILD does not influence the resulting distribution of categories in the resulting list. Furthermore, over-representation of certain categories is not explicitly punished by this metric.
The MMR method~\cite{carbonell1998use} approximately optimizes the ILD metric, by greedily growing a recommendation list $S$. The next item to be added to the recommendation list is chosen to be the one which maximizes the quantity $\lambda rel(u_i, v_k) + (1- \lambda) \min_{v_j\in S} dist(v_k,v_j)$, where $\lambda$ is a trade-off parameter between 0 and 1, and $rel(u_i,v_k)$ represents the relevance score of item $v_k$ to user $u_i$. \\

\textbf{Intent-Aware Expected Reciprocal Rank (ERR-IA):} The ERR-IA metric is the intent-aware version of Expected Reciprocal Rank metric, introduced by Chapelle et al \cite{chapelle2009expected}. ERR-IA considers the sum of each item category's weighted marginal relevance. To do so, we consider the quantity $p(R_a)$, which is the probability that the desired recommendation set's target category is $R_a$. Chapelle et al \cite{chapelle2009expected} formally define ERR-IA for some $u$'s given recommendation set $N(u_i) = \{v_j\}_{j=1}^{c_i}$ as follows (Again, $rel(v_k)$ denotes the relevance of item $v_k$ for this user):
\[\sum_{R_a \in \mathcal{R}} p(R_a) \sum_{k = 1}^{c_i}\frac{1}{k}rel(v_k)\prod_{\ell=1}^{k-1} (1 - rel(v_\ell)). \]

ERR-IA is a personalized metric and aims for good coverage of relevant categories in the recommendation list. However, it does not explicitly penalize the over-representation of a particular category provided that it is well-covered. This metric is optimized by the xQuAD reranking strategy~\cite{santos2013explicit}. Similar to the MMR method, xQuAD greedily optimizes for its metric by greedily picking items which maximize the marginal change in the ERR-IA metric plus a relevance term. \\


\textbf{Binomial Diversity (BD):}
Binomial diversity is a diversity measure due to Vargas et al \cite{vargas2014coverage}; we omit the complete description of this metric due to its intricacy. 
Roughly speaking, the authors use a binomial distribution to model the coverage and redundancy of the categories based on the items included in the recommendation list.
Binomial diversity punishes both the under-representation and over-representation of a given category in a user's list, and strives for a balance between coverage and non-redundancy. It can be optimized for in the same way as xQuAD optimizes for the ERR-IA metric, and a thorough experimental evaluation of this method is carried out by \cite{vargas2014coverage}. However, due to the complexity of the metric, no explicit guarantees can be given for the performance of the algorithm.

\subsection{Sales Diversity}
In addition to adding diversity to a single user's recommendation list, we are also interested in surfacing content for increased feedback from the users. Since it is impossible for the users to give feedback on items that are not surfaced adequately by the system, we measure our algorithms by two sales diversity measures. The first of these metrics is {\bf aggregate diversity}, which counts the number of items that are shown to at least one user \cite{adomavicius2011maximizing, ge2010beyond, adamopoulos2011unexpectedness, adomavicius2012improving}. Our thresholded item diversity objective can be thought of as a refinement of aggregate diversity, where each item needs to be recommended to multiple different types of users instead of just once to anyone in the system.
The second sales diversity measure that we employ in our experiments is the {\bf Gini index} which is also widely employed in the recommender system community~\cite{shani2011evaluating, ge2010beyond, herlocker2004evaluating, ren2014avoiding, carterette2009analysis}.
Category-aware metrics surveyed above try to solve the filter bubble problem for the users, while the type information can be used to solve the same problem for the business running the recommender system. In our work, we incorporate aggregate information symmetrically from both item-category and user-type information in our metrics to address this aspect.

\subsection{Graph-theoretic Approaches for Recommender Diversity}

The first paper to use a subgraph selection model for maximizing aggregate diversity is~\cite{adomavicius2011maximizing}, where the authors reduce the problem to bipartite matching. Other recent papers~\cite{antikacioglu2015recommendation,kdd17Version} have refined the notion of using subgraph selection by formulating more involved diversity metrics such as redundant coverage of items, and minimizing the discrepancy from a target degree distribution on the items, and solving for them using network flow and greedy techniques. Our paper follows this recent line of work.

\subsection{Submodularity and NP-Completeness}
The problem of maximizing a submodular set function has been extensively analyzed in the last 40 years, starting with Nemhauser et. al. \cite{nemhauser1978analysis} Many of the problems we pose reduce to maximizing such a function, which is  NP-hard even when the function is monotone increasing. Nonetheless, the problem can be approximated using the greedy algorithm, which gives a $(1-1/e)$-approximation in the simplest case. Moreover, the constraint of choosing a subset of a fixed size, which corresponds to a uniform matroid constraint, can be replaced by any other matroid constraint without affecting the approximation ratio \cite{abbassi2013diversity}.
Since the coverage type objectives we define in this \ifdefined\THESIS chapter \else paper \fi are submodular, and our main type of constraints form a partition matroid, we make extensive use of results in this area.
Other researchers have considered the use of submodular functions in diversifying recommendations, but only over the set of a single user's recommendation set~\cite{habibi2014enforcing, kuccuktuncc2013diversified,santos2010exploiting}.

\subsection{Summary of Contributions}
%
\begin{enumerate}[leftmargin=1em,labelwidth=*,align=left]
\item We introduce two new metrics for recommendation diversity that we call thresholded item-diversity ($TIDiv$) and thresholded user-diversity ($TUDiv$) which consider the distribution of user types among an item's recommended user set and the distribution of item categories in a user's recommendation item set respectively. $TUDiv$ is an objective that is similar to other category-aware diversity metrics, but $TIDiv$ is unique in considering diversity among an item's recommendees. TIDiv can be thought of as a sales diversity metric, and explicitly addresses the need of a business to collect feedback from different types of users.

\item In the case of disjoint types and categories, we model the problem of maximizing type diversity across all items and category diversity across all users as a subgraph selection problem. We reduce the resulting problem to a minimum cost flow problem and obtain exact polynomial time algorithms (Theorem ~\ref{thm:chapter3_main_disjoint} in Section~\ref{sec-disjoint}) .

\item We address the case of non-disjoint types and categories in Section~\ref{subsec:thresholded_submodular}, where we prove that the problem of maximizing the same objectives mentioned above is NP-complete (Theorem ~\ref{thm:np-completeness}). While this rules out an exact polynomial time solution, we obtain a $(1-1/e)$-approximation using the submodularity of our objectives. We also show how to modify the algorithm to run in nearly linear time in the number of candidate recommendations (Theorem ~\ref{thm:submodular_runtime}), making it very efficient.

\item We conduct experiments using the MovieLens dataset that consider both disjoint item categories and overlapping ones. Our experimental setup is described in Section~\ref{sec:experimental_setup}, and the results are in Section~\ref{exp_results}. We show that despite being flow based, our algorithms for the disjoint case can easily handle problems involving millions of candidate edges. We also show that the greedy algorithm we describe is competitive in efficiency with the reranking approaches we compare against (Subsection ~\ref{sec:chap3overlap}), and competitive with our optimal flow based approach when used with disjoint categories and types (Subsection ~\ref{sec:chap3disjoint}). Our algorithms perform better than the baselines across the board on sales diversity metrics, and obtain good values for the other intent-aware metrics despite only optimizing for them by proxy.

\end{enumerate} 


\section{Disjoint Types and Categories}
\label{sec-disjoint}

We model the problem of making recommendations as a subgraph selection problem on a bipartite graph $G = (L \cup R, E)$ where the partition $L$ represents a set of users and partition $R$ represents a set of items. For each user $u_i$, we have a space constraint $c_i$, which is due to display space limitations on a given webpage.
For each edge $(u_i,v_j)$ between user $u_i$ and item $v_j$ in $G$, we are also given a real-valued relevance $rel(u_i,v_j)$ that is typically an actual rating or a predicted rating from a CF system.
Often, the graph $G$ available for selection of recommendations is chosen by using a CF system's relevance scores and only retaining edges that are higher than a minimum threshold relevance or quality value.
In this section, we model the case when the subgroups of users and items are disjoint.

We define a collection of subsets $\mathcal{L} = \{L_1,L_2,...L_n\}$ on the user set $L$ that represent different types of users and are mutually disjoint. Similarly, we define a collection of subsets $\mathcal{R} = \{R_1,R_2,...,R_m\}$ on the item catalog which partition $R$ to represent different categories or genres of items. This means there exists functions $type: L \rightarrow \mathcal{L}$, which maps users to their designated type, and $cat: R \rightarrow \mathcal{R}$, which maps items to their corresponding category. The edges between users and items in $G$ represent possible recommendations that can be made. We wish to output a subgraph $H$ of recommendations where each user $u_i$ has $c_i$ recommendations.

\subsection{Global edge-wise diversity}
Consider a recommendation edge $(u_i,v_j)$ in the subgraph $H$.
Let $\delta_i^H(cat(v_j))$ denote the number of neighbors user $u_i$ has in $v_j$'s category and $\delta_j^H(type(u_i))$ denote the number of neighbors $v_j$ has in $u_i$'s type in $H$. In order to achieve a diverse set of recommendations, we would like each user to see a large number of categories, while also showing each item to a large number of user-types.
To define a diversity metric that takes both of these considerations into account, we consider assigning the following weight to each edge where $\beta$ and $\mu$ are real valued parameters
\[w_{ij} = \frac{\beta}{\delta_i^H(cat(v_j))}+\frac{\mu}{\delta_j^H(type(u_i))} .\]

A weighting like this is natural, since we are assigning less weight to recommendations that are not novel for either the user type or the item category that this recommendation serves.
For instance, a recommendation edge that gives the user the only item from a category, and the item the only user from a type, will have the maximum weight of $\beta+\mu$.
We can now define the diversity of a solution subgraph $H$ as follows.

\[Div_{\beta,\mu}(H) = \sum_{u_iv_j\in H} w_{ij}\]

and subsequently maximize this objective for a highly diverse set of recommendations.

\begin{figure}
\begin{boxedminipage}{\textwidth}
\textbf{Input:} A relevance-weighted bipartite graph $G(L,R,E)$, a vector of display constraints $\{c_i\}_{i=1}^l$, a collection of user types $\mathcal{L}$, a collection of item categories $\mathcal{R}$, real-valued parameters $\beta,\mu$ .\\
\textbf{Output:} A subgraph $H\subseteq G$, of maximum degree $c_i$ at each node $u_i\in L$, and maximizing the objective $Div_{\beta,\mu}(H) + rel(H)$.
\end{boxedminipage}
\caption{The definition of the $\text{MAX}-Div_{\beta,\mu}$ problem.}
\end{figure}

\begin{proposition}
With the definition above

\[ Div_{\beta,\mu}(H) = \beta \sum_{u_i\in L} | a : R_a\cap N(u_i) \ne \emptyset | + \mu \sum_{v_j\in R} | a : L_a\cap N(v_j) \ne \emptyset |.\]

\end{proposition}

\begin{proof}

We can think of each edge weight as a user contributing a fractional value towards the category the user is hitting as well as an item contributing a fractional value of towards the user-type the item gets hit by. For example, if a user $u_i$ has $4$ edges to some category, the value of each $\frac{\beta}{\delta_i^H(cat(v_j))}$ for every item $v$ in that category that $u$ is connected to is $\frac{\beta}{4}$. If some item $v_j$ has 3 edges coming from the same user-type, the value for $\frac{\mu}{\delta_j^H(type(u_i))}$ for each user $v_j$ is connected to is $\frac{\mu}{3}$. This means that, $Div(H)$ gets a value of $\beta$ for every category a user hits, and a value of $\mu$ for every user-type an item hits:
\begin{align*}
Div_{\beta,\mu}(H) &= \sum_{u_iv_j\in H} \frac{\beta}{\delta_i^H(cat(v_j))} + \frac{\mu}{\delta_j^H(type(u_i))} \\
     &= \sum_{u_i\in H} \sum_{R_a\cap N(u_i)} \frac{\beta}{|R_a\cap N(u_i)|} + \sum_{v_j\in H} \sum_{L_b\cap N(v_j)} \frac{\mu}{|L_b\cap N(v_j)|} \\
     &= \beta \sum_{u_i\in L} | a : R_a\cap N(u_i) \ne \emptyset | + \mu\sum_{v_j\in R} | a : L_a\cap N(v_j) \ne \emptyset |.
\end{align*}

\end{proof}

We can isolate both terms of this expression as their own objectives, which may be formalized as follows:

\[ UserDiv(H) = \sum_{u_i\in L} \sum_{R_a} 1[\exists v_j \in R_a: u_iv_j\in H]. \]

\[ ItemDiv(H) = \sum_{v_j\in R} \sum_{L_a} 1[\exists u_i \in L_a: u_iv_j \in H].\]

Here, $UserDiv(H)$ will give us reward proportional to the number of categories hit for each user and $ItemDiv(H)$ will give us reward proportional to the number of user types hit for each item.

 Ignoring type information, we first show that $UserDiv(H)$ can be optimized in polynomial time, since this construction is simpler to formulate and solve in practice.

\begin{theorem}
\label{user-diversity}
The problem of maximizing $Div_{\beta,\mu}$ can be reduced to a minimum cost flow problem if the categories are disjoint, i.e. $R_a\cap R_b = \emptyset$ for all $a,b$.
\end{theorem}
\begin{proof}
For each $u_i \in L$, we set supplies of $c_i$, and a demand of $\sum_{u_i\in L} c_i$ for a newly created sink node $t$. For each user $u_i$ and category $R_a$ such that $\exists v_j \in R_a$ such that $u_iv_j \in G$ we create nodes  $n_{i,a}$ and $n'_{i,a}$. We will create an arc of capacity 1 and cost $-1$ between every $u_i$ and $n'_{i,a}$. We will also add arcs of capacity $1$ and cost $0$ between every $n'_{i,a}$ and $n_{i,a}$ and arcs of unbounded capacity and cost $0$ between $u_i$ and $n_{i,a}$. For each edge $u_iv_j$ in $G$ where $v_j \in R_a$ we create an arc of capacity 1 and cost 0 between $n_{i,a}$ and $v_j$. Finally, from each $v_j \in R$ we make an arc of unbounded capacity and cost 0 to the sink node $t$.

We let the solution subgraph $H$ be the subgraph of $G$ formed by using edges $u_iv_j$ for all arcs of the form ($n_{i,a}$,$v_j$) used in the flow. Each node now gets to take one recommendation in each new category for a cost of $-1$, Therefore, the cost of a flow defined by $H$ is $ - \sum_{u_i\in L} \sum_{R_a} 1[\exists v_j \in R_a: u_iv_j\in H]$. Minimizing this quantity is the same as maximizing $UserDiv(H)$, which proves the result.
\end{proof}

\begin{proposition}
If every user is his own type, then subject to display constraints, $Div_{\beta,\mu}(H) \propto UserDiv(H)$, and $Div_{\beta,\mu}(H)$ can be maximized exactly in polynomial time.
\end{proposition}
\begin{proof}
If every user is his own type, then the quantity $| a : L_a\cap N(v_j) \ne \emptyset |$ simply counts the number of edges incident on an item $v_j$. Therefore, we obtain

\begin{align*}
Div_{\beta,\mu}(H) &= \beta \sum_{u_i\in L} | a : R_a\cap N(u_i) \ne \emptyset | + \mu \sum_{v_j\in R} \delta^H(v_j) \\
     &= \beta UserDiv(H) + \mu \sum_{u_i\in L} \delta^H(u_i) \\
     &= \beta UserDiv(H) + \mu \sum_{u_i\in L} c_i .
\end{align*}

Since the quantity on the right is constant, the result follows from Theorem \ref{user-diversity}.
\end{proof}

Finally, we prove the theorem in the most general case by combining the objectives $UserDiv(H)$ and $ItemDiv(H)$. In fact, this is possible while incorporating rating relevance into the objective. In particular, let $rel(u_i,v_j)$ denote the relevance of item $v_j$ to user $u_i$. Then the relevance based quality of the entire recommender system can be computed as $rel(H) = \sum_{(u_i,v_j) \in H}  rel(u_i,v_j)$. We can now state the main result of this section.

\begin{theorem}
\label{thm:u_plus_i}
The MAX-$Div_{\beta,\mu}$ problem can be reduced to a minimum cost flow problem if both user types and item categories are disjoint, i.e. $R_a \cap R_b = \emptyset$ and and $L_a \cap L_b = \emptyset$ for all $a,b$.
\end{theorem}

We omit the proof in favor of presenting a more general result in Theorem \ref{thm:chapter3_main_disjoint}.


\subsection{Diversity Thresholds}
\label{subsec:thresholded_disjoint}

While increasing user and item diversity is important, one downfall to our method is that it fails to take into account the fact that the relevance of each category to a user may be different. It may not be beneficial for our recommender to show a user items from every different category possible, since that user may not be interested in some of those categories to begin with. The same can be said for the item side: item diversity may increase an item's popularity and help it collect feedback, however, an item should be shown to users in its target audience more than users outside its target audience.

To fix this, and help guide our algorithm to selecting more relevant recommendations for each user and item,
we propose setting diversity thresholds for each user-category and item-type pair. For categories that the user cares a lot about, we can increase this threshold while setting it to zero for those that the user is not interested in at all. Let us denote $\rho_i(R_a)$ as user $u_i$'s threshold for recommendations made to items in category $R_a$, and $\lambda_j(L_b)$ be an item $v_j$'s threshold for recommendations made from users of type $L_b$. We now define two updated objectives that take these thresholds into account:

\[TUDiv(H) = \sum_{u_i \in L} \sum_{R_a} \min(\rho_i(R_a), \delta_i^H(R_a)). \]

\[TIDiv(H) = \sum_{v_j \in R} \sum_{L_b} \min(\lambda_j(L_b),\delta_j^H(L_b)). \]

Notice that relative to $UserDiv$, for a user $u_i$, we are simply increasing the diversity gain from a  seeing new items from a category $R_a$ up to a threshold value of $\rho_i(R_a)$. If we set all the $\rho$ and $\lambda$ values to 1 in the above expressions, we recover $UserDiv$ and $ItemDiv$.

We can again consider these two objectives together to form a single objective that will maximize the thresholded diversity of a solution subgraph $H$, where $\beta$ and $\mu$ are real-valued parameters:
\[ TDiv_{\beta,\mu}(H) = \beta \cdot TUDiv(H) + \mu \cdot TIDiv(H) .\]

\begin{figure}
\begin{boxedminipage}{\textwidth}
\textbf{Input:} A weighted bipartite graph $G(L,R,E)$, a vector of display constraints $\{c_i\}_{i=1}^l$, a collection of user types $\mathcal{L}$, a collection of item categories $\mathcal{R}$, user-category thresholds $\{\rho_i(R_a)\}_{i,a}$, item-type thresholds $\{\lambda_j(L_b)\}_{j,b}$, real-valued parameters $\beta,\mu$ .\\
\textbf{Output:} A subgraph $H\subseteq G$, of maximum degree $c_i$ at each node $u_i\in L$, and maximizing the objective $TDiv_{\beta,\mu}(H) + rel(H)$.
\end{boxedminipage}
\caption{The definition of the $\text{MAX}-TDiv_{\beta,\mu}$ problem.}
\end{figure}

The main result of this section is that in the case where user types and item categories are disjoint, $TDiv(H)$ can still be optimized in polynomial time.

\begin{theorem}

The $\text{MAX}-TDiv_{\beta,\mu}$ problem can be reduced to a minimum cost network flow problem if both user types and item categories are disjoint, i.e. $R_a \cap R_b = \emptyset$ and and $L_a \cap L_b = \emptyset$ for all $a,b$.
\label{thm:chapter3_main_disjoint}
\end{theorem}

\begin{proof}

A diagram of the construction can be found in Figure ~\ref{fig:chapter3_tudiv_plus_tidiv_diagram}.
Our network will have nodes for all users $u_i \in L$ and items $v_j \in R$ of $G(L\cup R, E)$, and a sink node $t$. The supply for each user $u_i$ will be its corresponding space constraint $c_i$. For every category $R_a$ that a user $u_i$'s recommendations hit, we will create two nodes $n_{i,a}$ and $n_{i,a}'$. Let there be an arc of capacity $\rho_i(R_a)$ and cost $-\beta$ between $u_i$ and $n_{i,a}'$ and an arc with capacity $\rho_i(R_a)$ and cost $0$ between $n_{i,a}'$ and  $n_{i,a}$. There will also be an arc with unbounded capacity and cost $0$ between $u_i$ and $n_{i,a}$. Similarly, for an item $v_j$, we will create two nodes for each user type   $L_b$ its incoming edges are from, $m_{j,b}$ and $m_{j,b}'$. Let there be an arc of capacity $\lambda_j(L_b)$ and cost $-\mu$ between $m_{j,b}'$ and $v_j$, and an arc of capacity $\lambda_j(L_b)$ and cost $0$ between $m_{j,b}$ and $m_{j,b}'$. We will also add an arc of unbounded capacity and cost $0$ between $m_{j,b}$ and $v$. For each edge $(u_i,v_j) \in E$, where $v_j$ is in category $R_a$ and $u_i$ is of type $L_b$, we will add an arc with cost $-rel(u_i,v_j)$ and capacity $1$ between $n_{i,a}$ and $m_{j,b}$. Finally, there will be an arc from every item $v_i$ to the sink $t$ with unbounded capacity and cost $0$.

We let the solution subgraph $H$ be the subgraph of $G$ formed by using edges ($u_i,v_j$) for all arcs of the form ($n_{i,a}$,$m_{j,b}$) used in the flow. The cost of the flow induced by $H$ will therefore cost \[-\beta \sum_{u_i \in L} \sum_{R_a} \min(\rho_i(R_a), \delta_i^H(R_a)) - \mu\sum_{v_j \in R} \sum_{L_b} \min(\lambda_j(L_b),\delta_j^H(L_b)) - rel(H)\] since we may use the $-\beta$ cost arc for each user-category pair and the $-\mu$ cost arc item-type pair up until they reaches capacity. This quantity is simply \[-\beta TUDiv(H) - \mu TIDiv(H) -rel(H) = -TDiv_{\beta,\mu}(H) -rel(H).\]
Therefore, minimizing this quantity will maximize $TDiv_{\beta,\mu}(H) + rel(H)$.
\ \\
\end{proof}

\begin{figure}
\centering
 \includegraphics[width=0.95\columnwidth]{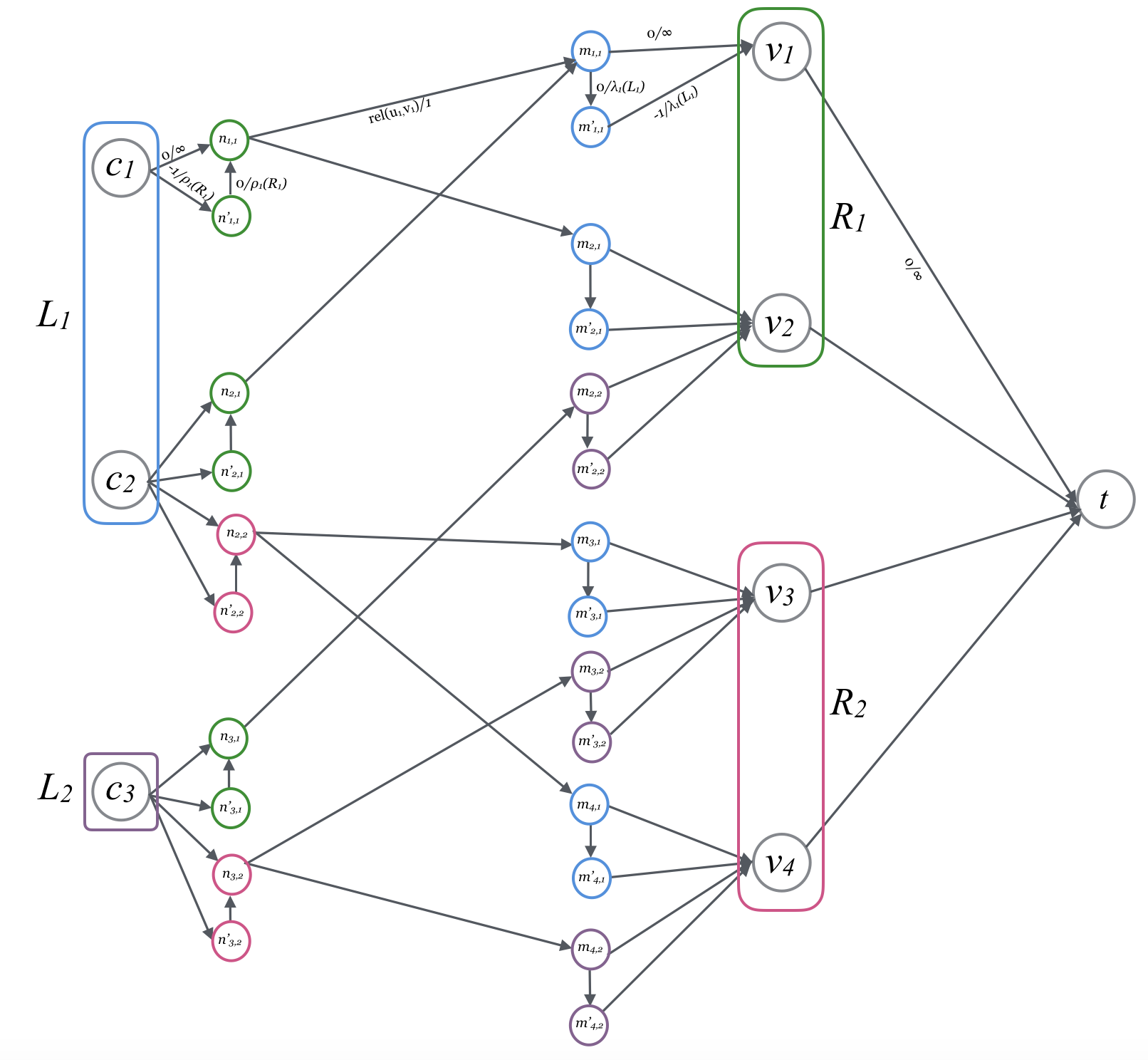}
\caption{Construction of the flow problem in Theorem ~\ref{thm:chapter3_main_disjoint}. \label{fig:chapter3_tudiv_plus_tidiv_diagram}}
\end{figure}


Our results about disjoint categories and type are useful in applications such as news recommendation where users are split into natural categories according to their political alignment, and the news articles and their publishers are split according to the same categorization. However, these results can be applied without any modification to other domains such as retail, where the products (items) are split into natural retail categories according to product ontologies, and where the users are split according to natural mutually exclusive demographic types such as gender or age and income brackets. However, the more general case is the one when categories and items are not necessarily disjoint which we turn to next. 
\begin{figure}	
\centering
 \includegraphics[width=\columnwidth]{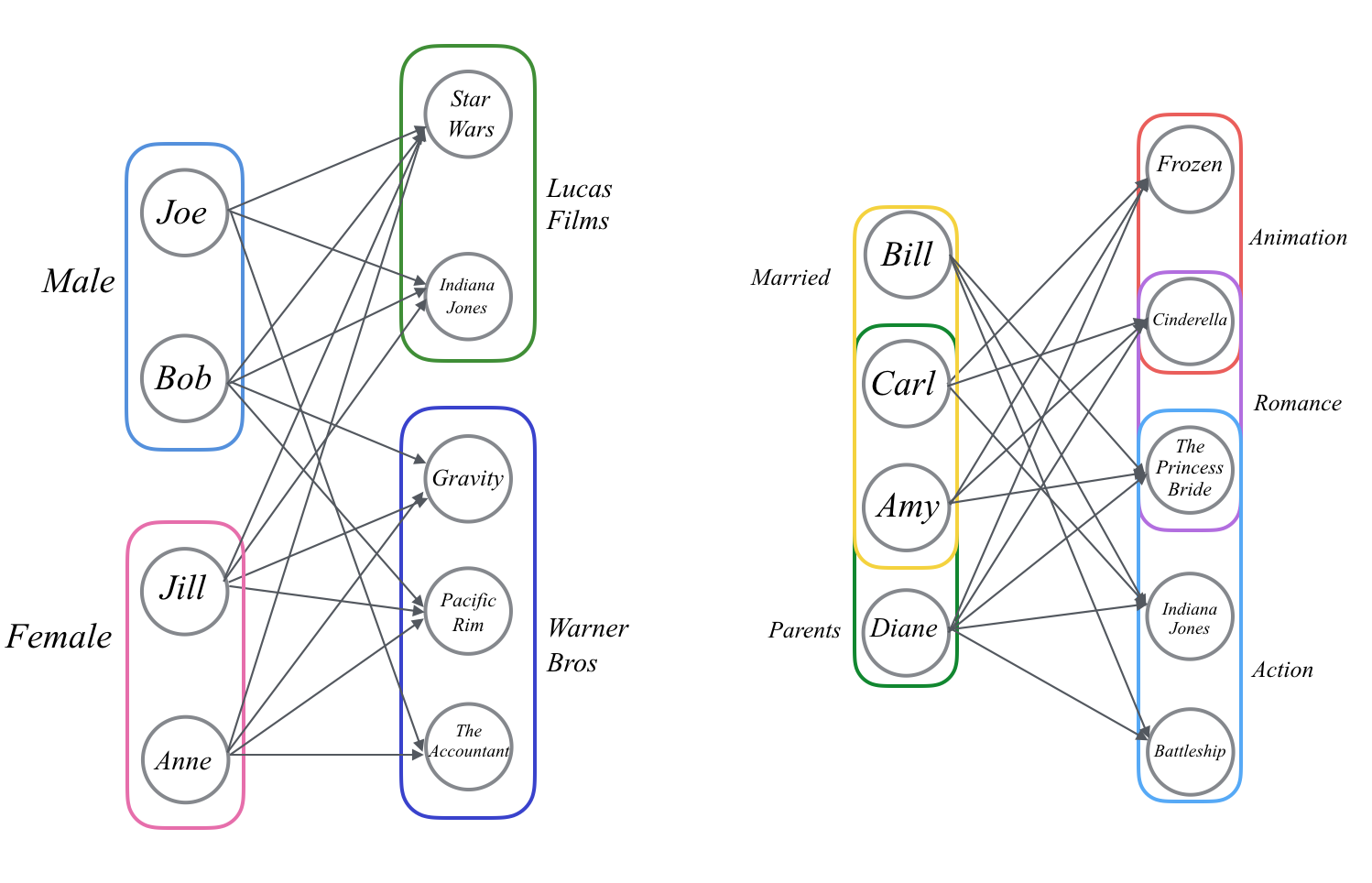}
\caption{(Left) Example showing disjoint user categories (by gender) and disjoint movie types (by production company). (Right) An example showing overlapping user categories by demographic features and overlapping movie types by genre \label{fig:disj}}
\end{figure}

\section{Overlapping Types and Categories}
\label{subsec:thresholded_submodular}

Although cases involving disjoint user-types and categories are solvable in polynomial time, in actual practice, categories of items are not necessarily disjoint, and users may be assigned to more than one user-type. 
We continue to use the notation from Section~\ref{sec-disjoint}, but now user types $\mathcal{L} = \{L_1,L_2,...L_n\}$ on the user set $L$ may overlap as well as  item categories $\mathcal{R} = \{R_1,R_2,...,R_m\}$ on the catalog $R$.
When item categories and user types are non-disjoint, maximizing $TDiv(H)$ is NP-Hard, which can be seen in the following theorem (via a simple reduction from max-coverage).

\begin{theorem}
\label{thm:np-completeness}
Finding an optimal solution to maximize $TDiv_{\beta,\mu}(H)$ with non-disjoint categories and types is NP-Hard.
\end{theorem}
\begin{proof}
We fix $\beta = \mu = 1$ since proving the NP-Hardness of a special case is sufficient. We show that optimizing just $TUDiv_{\beta,\mu}(H)$ with a single user is NP-Hard with the following reduction from the Max-Cover problem, a well known NP-Hard problem: Given a set of elements $\{1,2,3,...,n\}$, a collection of $m$ sets $\mathcal{S}$, and an integer $k$, we want to find the largest number of elements covered by at most $k$ sets.

We construct a bipartite graph $G(L\cup R, E)$ where $|L| = 1$ and $|R| = m$ with the items representing sets. The vertex $u \in L$ has an out-degree of $m$, one for each vertex in $R$. We then create subsets $R_1,R_2,R_3,...,R_n \subseteq R$ with one such subset corresponding to each element $e_i$ in $\{1,2,3,...,n\}$: the subset of vertices in $R_i$ correspond to the sets in $\mathcal{S}$ that contain the element $e_i$. We also set $\{\rho_1(R_i)\}_{i=1}^n = 1$.
We let $c_1 = k$. An optimal solution for $TUDiv(H)$ would give an optimal solution to Max-Cover, since finding the maximum number of categories hit with $k$ edges out of $L$ would find the maximum number of elements we can cover with $k$ sets.

Since finding the optimal solution to $TUDiv_{\beta,\mu}(H)$ is NP-Hard, and $TUDiv(H)$ is a special case of $TDiv_{\beta,\mu}(H)$, we have shown that optimizing $UserDiv(H)$ will also be NP-Hard, thus proving the desired result.
\end{proof}

\begin{algorithm}
 \caption{The greedy algorithm for $TIDiv$ and $TUDiv$ maximization}
\label{alg:chapter3greedy}
\begin{algorithmic}
\STATE \textbf{Data:} A bipartite graph $G=(L,R,E)$ and display constraint $c$
  \STATE \textbf{Result:} A solution graph $H$ maximizing $TDiv_{\beta,\mu}(H)+rel(H)$
 \WHILE{some vertex $u_i\in L$ has $deg_H(u_i)< c_i$}
\STATE  $(u_i,v_j) = e \longleftarrow \arg \max_{e'\in E} ~ TDiv_{\beta,\mu}(H\cup\{e'\})- TDiv_{\beta,\mu}(H) + rel(e')$\;
  \IF {$deg_H(u_i) < c$} 
  \STATE $H \gets H\cup\{e\}$ \ENDIF \ENDWHILE
  \RETURN H;
\end{algorithmic}
\end{algorithm}

Since we are not able to maximize $TDiv_{\beta,\mu}(H)$ optimally, we can make use of the fact that $TDiv_{\beta,\mu}(H)$ is both monotone and submodular, which will allow us to apply a greedy algorithm which will yield a $(1-1/e)$-approximation ratio.

\begin{proposition}
$TUDiv(H)$ is submodular.
\end{proposition}
\begin{proof}
Let $X$ and $Y$ be two sets of edges such that $X \subseteq Y$ and let $e$ be an edge not in $X$ or $Y$. Consider the quantity $TUDiv(X \cup \{e\}) - TUDiv(X)$. Observe that this is the number of categories $R_a$ that $e$ will saturate (not including categories that have already reached their threshold). This will be at least as much as the number of categories $e$ saturates in $Y$, since $Y$ could contain edges that have already saturated categories that $e$ would saturate. It follows that $TUDiv(X \cup \{e\}) - TUDiv(X) \geq TUDiv(Y \cup \{e\}) - TUDiv(Y)$. This satisfies the ``diminishing returns" property of submodular functions. Therefore $TUDiv(H)$ is submodular.
\end{proof}

We get the following from a symmetric argument.
\begin{proposition}
$TIDiv(H)$ is submodular.
\end{proposition}

\ifdefined\THESIS
\begin{proof}
We again consider edge sets $X$ and $Y$ where $X \subseteq Y$ and let $e$ be an edge not yet in $X$ or $Y$. Observe that the quantity $TIDiv(X \cup \{e\}) - TIDiv(X)$ counts the number of types $L_i$ that $e$ saturates (not counting user types that have already reached their threshold). This will be at least as much as the number of types $e$ saturates in $Y$, since $Y$ could contain edges that saturate types that $e$ saturates. This means that $TIDiv(X \cup \{e\}) - TIDiv(X) \geq TIDiv(Y \cup \{e\}) - TIDiv(Y)$. Therefore, $TIDiv(H)$ is submodular.
\end{proof}
\fi

\begin{corollary}
$TDiv_{\beta,\mu}(H)$ is submodular.
\end{corollary}
\ifdefined\THESIS
\begin{proof}
Since both $TUDiv(H)$ and $TIDiv(H)$ are submodular, and non-negative linear combinations of submodular functions are also submodular, $TDiv(H)$ will also be submodular.
\end{proof}
\fi

\begin{corollary}
The objective function $rel(H)$ is submodular.
\end{corollary}

\ifdefined\THESIS
\begin{proof}
The sum of the relevance values of all the recommendations is a linear function. Linear functions are modular and hence also submodular, and submodular functions are closed under addition, so this function is submodular as well.
\end{proof}
\fi

The monotonicity and the submodularity of the objective now allows us to write the greedy algorithm given in Algorithm 1.

Stated in its current form, the greedy algorithm takes $O(|E|^2)$ to run. However, it is possible to speed it up significantly by using better data structures. 

\begin{theorem}
\label{thm:submodular_runtime}
Let $R_1, \ldots, R_k$ be the set of overlapping categories and $L_1, \ldots, L_p$ be the set of overlapping types for a $TDiv$ maximization problem. Then the greedy algorithm can be implemented to run in time, $O((E+\sum_{a=1}^k R_a +\sum_{b=1}^p L_b)\log |E|)$.
\end{theorem}
\begin{proof}
Let $u\in L$, $v\in R$, and let $u,v \in G$ be a candidate recommendation. The category contribution of this edge to a partial solution $H$ is the number of categories $R_i$ that $v$ belongs to, for which $\rho(R_i) < \delta_L^H(R_i)$ is satisfied. Similarly, the type contribution of this edge is the number of types $L_j$ that $u$ belongs to, for which  $\lambda(L_j) < \delta_R^H(L_j)$. While constructing the solution, both of these quantities can only decrease. Furthermore, we are only ever interested in the node with the highest marginal contribution.

Therefore, we can keep track of the potential contribution of each edge in a max-heap. Initially, the priority of each edge is set to be the number of categories and number types it covers. Each time an edge meets a category target, we decrease the priority of every unused edge incident on that category by $\beta$. Similarly, when a user type target is satisfied, we decrease the priority of every unused edge incident on that type by $\mu$. Both operations take logarithmic time using a heap which supports the decrease-key operation. This operation is performed at most once for each type and category.

This means that we are maintaining a max-heap with $|E|$ elements, removing the maximal element $|E|$ times, and decreasing the key of some edge by at most $\sum_{a=1}^k R_a +\sum_{b=1}^p L_b$ times. Both of these operations can be done in $O(\log |E|)$ time, which gives us the desired runtime.
\end{proof} 
 f\section{Experimental Setup}
\label{sec:experimental_setup}

\subsection{Datasets}
\textbf{Category Data:} We use ratings data as well as type and category data from the MovieLens-1m dataset, and additional category data from IMDB. For disjoint user types in the MovieLens dataset, we use three different demographic data points included in the data: age-group (6 different values), gender (2 different values), and occupation (19 different values) each of which form a partition the user set.

\textbf{Supergraph Generation:}
We used the MovieLens-1m \cite{MLData} rating dataset to generate the graph we fed to our algorithms. (The data set can be downloaded \href{http://grouplens.org/datasets/movielens/1m/}{here.}) We pre-processed the dataset to ensure that every user and every item has an adequate amount of data on which to base predictions. This post processing left the MovieLens-1m data with 5800 users and 3600 items.
The use of this dataset is standard in the recommender systems literature. In this \ifdefined\THESIS chapter \else work \fi, we consider the rating data to be triples of the form $(user,item,rating)$, and discard any extra information.

Each dataset was processed in two different ways, once for experiments involving disjoint categories and once for experiments involving overlapping categories. In each case, the full dataset was filtered for items for which the category information was known. Each of these were then split 5-ways into holdout test sets and training sets. Only users for which more than 50 ratings were considered for inclusion in the test, and we denote this set of users by $L_T\subseteq L$. The training sets were then fed into a matrix factorization algorithm due to Hu et al. \cite{hu2008collaborative} with 50 latent factors. We set the input confidence value parameter $\alpha$ in their method to the value of 40, as recommended by the authors, and performed a grid search for the best regularization parameter $\lambda$ using 5-fold cross validation. Using the resulting user and item factor matrices, for each user we predicted the ratings of all the items for which the user did not provide feedback in the training test. Among these predicted ratings, we retained the 250 highest rated items along with their predicted ratings to feed into our algorithms.

\subsection{Quality evaluation}

We measure the effectiveness of our algorithms and others' along several orthogonal dimensions. For relevance, we report precision values, i.e. the fraction of items in the recommendation set that match items given in the test set. Formally, if we denote the set of recommendations given to a user in subgraph $H$ as $N(u_i)$ and the set of relevant held-out items for the user as $T(u_i)$, we define precision as follows.
\[ P = \frac{1}{|L_T|}\sum_{u_i \in L_T} \frac{|N(u_i) \cap T(u_i)|}{c_i}\]

In this paper, a held-out item is considered to be relevant to a user in our evaluation if its assigned rating was 3 or higher. 
We note however that this notion of precision is conditioned on the inherent diversity already represented by the ratings in the MovieLens-1M database, and hence may not be ideal. 
Therefore, aside from relevance based metrics, we also report two sales diversity metrics: aggregate diversity and the Gini index.

\ifdefined\THESIS
The definitions of these metrics can be found in ~\ref{Chapter3}.
\else
Aggregate diversity is simply the fraction of items in the catalog which have been recommended to at least one user, and it measures coverage. The Gini index measures how inequitable the recommendation distribution is. More concretely, if the degree distribution of the items is given as a sorted list $\{d_i\}_{i=1}^r$, then the Gini index is defined as follows.
\[ G = 1-\frac{1}{r}\left( r+1-2\frac{\sum_{i=1}^r (r+1-i)d_i}{\sum_{i=1}^r d_i}\right) \]
\fi

Finally, we report the objectives for which our methods explicitly optimize. These are ERR-IA for the xQuAD reranker, ILD for the MMR reranker, Binomial Diversity for the Binomial Diversity reranker. Among these, only the Binomial Diversity reranking method takes a parameter $\alpha$, which corresponds to a personalization parameter. The authors use the value $\alpha=0.5$ in their experimental evaluation \cite{vargas2014coverage}, and we also use this setting. We measure each of these metrics as well as our own $TUDiv$ and $TIDiv$ as they are measured among only the relevant items in the test set.

As mentioned in Subsection ~\ref{subsec:thresholded_disjoint}, for the $TUDiv$ and $TIDiv$ metrics, we set the thresholds using the training data. In particular, for the case of disjoint categories, we count the number of times each category appears in a user's training set, normalize these values to sum to the display constraint, and round to integer values. For the case of overlapping categories, we perform the same operation, but normalize the thresholds to sum to the display constraint times the average number of categories for an item in the training set. In the case of disjoint types, we again set the type thresholds proportional to the distribution of types found in the training data, but normalize the distribution to sum to 20\% of the average number of recommendations an item would have received if every item were equally promoted by the recommender system. This allows the measure to promote sales diversity among items, while respecting its interaction history with the users.

Note that setting the thresholds using the proportions in the training data inherently biases the distribution to which we are targeting the final diversity to, and makes it match the distribution in the overall training set. 
Despite this, we chose these thresholds so they match the precision measure that we use to evaluate the effectiveness of our methods.
Note that these thresholds can also be set by a designer who prefers to move the proportions of different categories for different users in a different direction that is found in the training data (and symmetrically for the items), but it would be difficult to evaluate the effectiveness of the resulting lists against other methods.

In our tables, we abbreviate the names of these metrics as $P$ for Precision, $A$ for aggregate diversity, $G$ for the Gini index, $BD$ for Binomial Diversity, and $ILD$ for intra-list distance. The number next to each metric denotes the cutoff at which it was evaluated.

\vspace{-.2cm}

\subsection{Baselines}
We compare our method against 3 baselines methods: the Binomial Diversity reranker due to Vargas et al. \cite{vargas2014coverage}, the MMR reranker due to Carbonell et. al \cite{carbonell1998use}, and the xQuAD algorithm due to Santos et al. \cite{santos2013explicit}. Each method takes a parameter $\lambda\in[0,1]$ which trades off relevance with the metric which is being optimized. For each of these methods, we performed a grid search for the best trade-off parameter, and report all the measurements for the setting which produced the best results for the method's corresponding metric. Since our algorithms have two trade-off parameters $\mu$ and $\beta$ in the objective $rel(H) + TDiv_{\beta,\mu}$ corresponding to two different metrics, we perform a grid search along both dimensions and report the two solutions which maximize $TIDiv$ and $TUDiv$ respectively. We additionally report the same metrics for the undiversified recommendation lists provided by the matrix factorization method under the heading ``TOP''.

\vspace{-.2cm}

\subsection{Software}
For the matrix factorization based recommender we trained, we used the implementation of Hu's matrix factorization method found in Ranksys \cite{sandoval2015novelty}. The baseline methods we compare against are also implemented in the same library. Our methods and metrics were implemented in a way to be compatible with the same library.
Additionally, we used a minimum cost network flow optimizer written by Bertolini and Frangioni~\cite{Frangioni2010MCFSimplex}.
The code we used for our experiments can be found in {\href{https://github.com/antikacioglu/salesdiversity/tree/master/Category}{our  repository}.
\vspace{-.2cm} 
\section{Experiments}
\label{exp_results}
In this section we report our findings on diversifying recommendations in MovieLens derived recommendation problems. \ifdefined\RECSYS We omit the results from the Netflix dataset for lack of space. \fi Our findings can be summarized as follows:

\ifdefined\THESIS
\begin{enumerate}
\else
\begin{enumerate}[leftmargin=1em,labelwidth=*,align=left]
\fi
\item In the setting of overlapping item categories, the greedy algorithm leveraging the submodularity of the $TDiv$ objective obtains significant gains in the $TIDiv$ and $TUDiv$ recommendation diversity metrics. Our algorithm preserves or improves the accuracy of the baseline recommender system, while also increasing sales diversity metrics.

\item In the setting of disjoint item categories, we show that the flow based algorithm obtains solutions which have higher predictive accuracy and higher sales diversity measurements. However, the differences are small enough for the greedy algorithm to make a suitable replacement for the more expensive, flow-based optimization technique.

\item The greedy algorithm is faster than competing diversification techniques, making it suitable for large scale recommendation tasks, provided that the heap used in its implementation can fit in memory.

\end{enumerate}

\subsection{Experiments on Overlapping Categories}
\label{sec:chap3overlap}

\begin{sidewaystable}
\begin{center}
\begin{tabular}{ |c|c|c|c|c|c|c|c|c| }
\hline
Method                           & P@20 & ERR-IA@20 & ILD@20 & BD@20 & TUDiv@20 & TIDiv@20 & A@20 & G@20 \\ \hline
TOP                              & 0.244 & 0.191 & 0.157 & 0.203 & 0.169 & 0.178 & 0.386 & 0.082 \\ \hline
xQuAD    ($\lambda=0.4$)         & \textbf{0.264} & \textbf{0.236} & 0.165 & 0.224 & 0.207 & 0.139 & 0.368 & 0.072 \\ \hline
MMR      ($\lambda=0.4$)         & 0.233 & 0.183 & \textbf{0.172} & 0.227 & 0.158 & 0.139 & 0.371 & 0.075 \\ \hline
BD       ($\lambda=0.6$)         & 0.227 & 0.208 & 0.156 & \textbf{0.255} & 0.152 & 0.176 & 0.384 & 0.084 \\ \hline
Greedy   ($\beta=0.4, \mu=4.0$)  & 0.252 & 0.201 & 0.163 & 0.238 & \textbf{0.210} & 0.189 & 0.420 & 0.087 \\ \hline
Greedy   ($\beta=4.0, \mu=0.2$) & 0.245 & 0.200 & 0.158 & 0.237 & 0.203 & \textbf{0.214} & \textbf{0.559} & \textbf{0.107} \\ \hline
\end{tabular}
\caption{MovieLens diversifications for artistic genre based categories and age group based types. The best value in each metric is bolded. \label{tbl:chapter3_overlap_table1}}
\smallskip

\begin{tabular}{ |c|c|c|c|c|c|c|c|c| }
\hline
Method                             & P@20 & ERR-IA@20 & ILD@20 & BD@20 & TUDiv@20 & TIDiv@20 & A@20 & G@20 \\  \hline
TOP                                & 0.244 & 0.191 & 0.157 & 0.203 & 0.169 & 0.148 & 0.386 & 0.082 \\ \hline
xQuAD    ($\lambda=0.4$)           & \textbf{0.264} & \textbf{0.236} & 0.165 & 0.224 & 0.207 & 0.139 & 0.368 & 0.072 \\ \hline
MMR      ($\lambda=0.4$)           & 0.233 & 0.183 & \textbf{0.172} & 0.227 & 0.158 & 0.139 & 0.371 & 0.075 \\ \hline
BD       ($\lambda=0.6$)           & 0.227 & 0.208 & 0.156 & \textbf{0.255} & 0.152 & 0.144 & 0.384 & 0.084 \\ \hline
Greedy   ($\beta=0.4, \mu=4.0$)    & 0.253 & 0.201 & 0.163 & 0.238 & \textbf{0.210} & 0.157 & 0.412 & 0.086 \\ \hline
Greedy   ($\beta=4.0, \mu=0.2$)    & 0.246 & 0.201 & 0.159 & 0.237 & 0.204 & \textbf{0.181} & \textbf{0.545} & \textbf{0.104} \\ \hline
  \end{tabular}
\caption{MovieLens diversifications for artistic genre based categories and occupation based types. The best value in each metric is bolded. \label{tbl:chapter3_overlap_table2}}
\smallskip

\begin{tabular}{ |c|c|c|c|c|c|c|c|c| }
\hline
Method                             & P@20 & ERR-IA@20 & ILD@20 & BD@20 & TUDiv@20 & TIDiv@20 & A@20 & G@20 \\ \hline
TOP                                & 0.244 & 0.191 & 0.157 & 0.203 & 0.169 & 0.197 & 0.386 & 0.082 \\ \hline
xQuAD    ($\lambda=0.4$)           & \textbf{0.264} & \textbf{0.236} & 0.165 & 0.224 & 0.207 & 0.139 & 0.368 & 0.072 \\ \hline
MMR      ($\lambda=0.4$)           & 0.233 & 0.183 & \textbf{0.172} & 0.227 & 0.158 & 0.139 & 0.371 & 0.075 \\ \hline
BD       ($\lambda=0.6$)           & 0.227 & 0.208 & 0.156 & \textbf{0.255} & 0.152 & 0.195 & 0.384 & 0.084 \\ \hline
Greedy   ($\beta=0.4, \mu=4.0$)    & 0.253 & 0.200 & 0.163 & 0.238 & \textbf{0.210} & 0.210 & 0.423 & 0.087 \\ \hline
Greedy   ($\beta=4.0, \mu=0.2$)    & 0.243 & 0.202 & 0.156 & 0.236 & 0.202 & \textbf{0.233} & \textbf{0.564} & \textbf{0.111} \\ \hline
\end{tabular}
\caption{MovieLens diversifications for artistic genre based categories and gender based type data. The best value in each metric is bolded. \label{tbl:chapter3_overlap_table3}}
\end{center}
\end{sidewaystable}

\ifdefined\THESIS
\begin{figure}	
\centering
 \includegraphics[width=\textwidth]{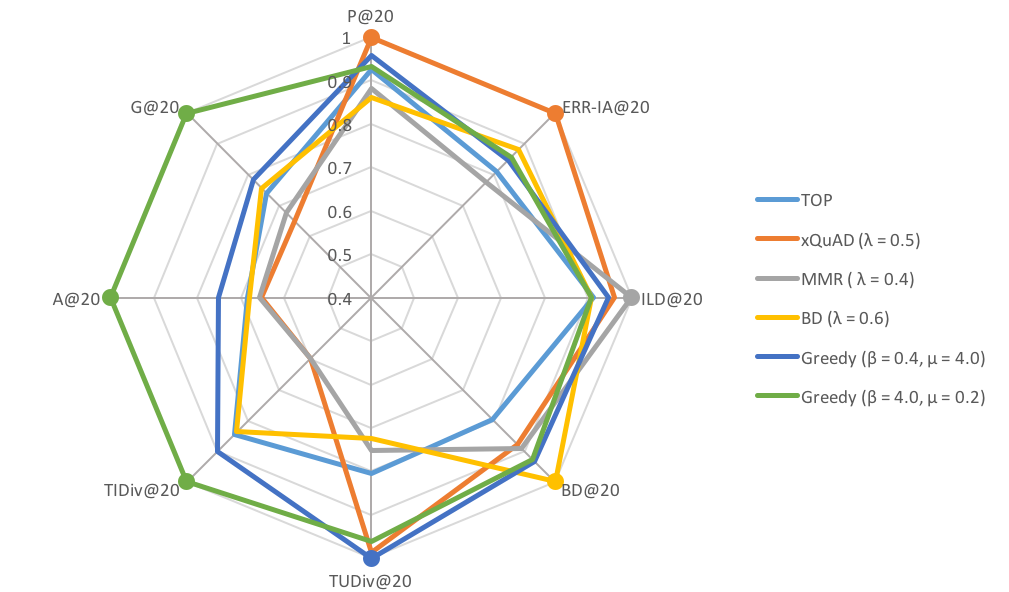}
\caption{A radial graph showing the relative performance of the reranking methods we tested for MovieLens data and gender based diversification. \label{fig:chapter3_overlap_radial}}
\end{figure}
\else
\begin{figure}
\centering
 \includegraphics[width=0.95\columnwidth]{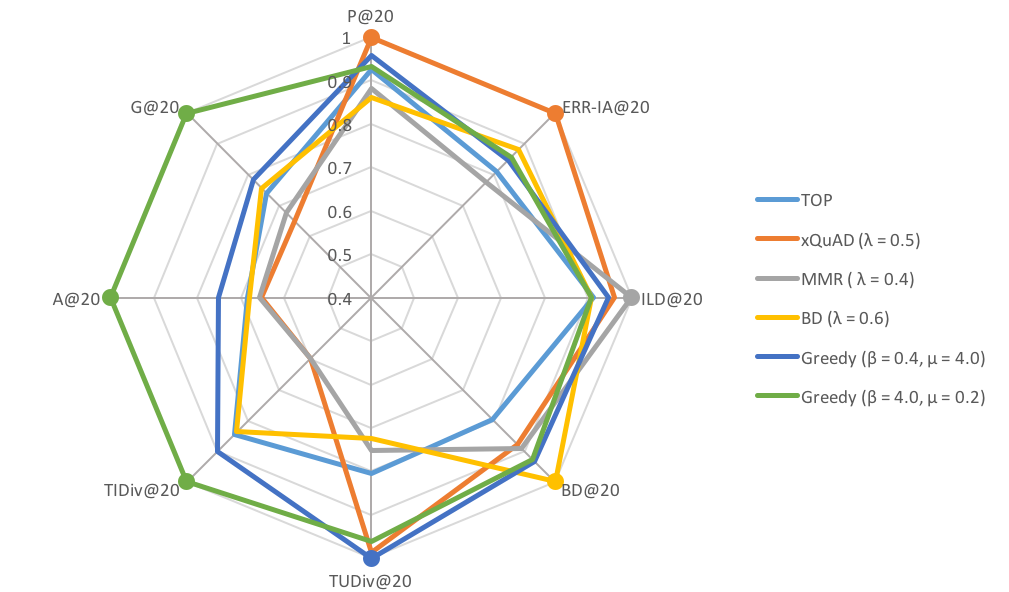}
\caption{A radial graph showing the relative performance of the reranking methods we tested for MovieLens data with movie genre and gender based diversification. \label{fig:chapter3_overlap_radial}}
\end{figure}
\fi

We first present our experiments on overlapping categories based on the artistic genre information for the movies, and user types based on age groups, occupation and gender respectively. Our results are summarized in Tables ~\ref{tbl:chapter3_overlap_table1},\ref{tbl:chapter3_overlap_table2},\ref{tbl:chapter3_overlap_table3}.
The relative performance of the methods we tested for artistic genre categories for the movies and genders of the users can be seen in Figure \ref{fig:chapter3_overlap_radial}.  
As expected each diversification method is best at maximizing its own objectives. In the case of our methods, this is true for both $TIDiv$ and $TUDiv$. Among the metrics we tested, both our greedy algorithm and the xQuAD algorithms made minor improvements to the precision of the recommendation lists, while Binomial Diversity and MMR slightly deteriorated the precision values. However, these differences are minor and each algorithm was able to find a good trade-off between relevance and diversity under suitable parameter settings.

Among the intent-aware metrics we have tested, our algorithms provide a very good proxy for Binomial Diversity, while performing less well on the Intra-List Distance and ERR-IA metrics.  This can be explained by the fact that Binomial Diversity, unlike the ERR-IA and ILD metrics, explicitly penalizes redundancy. Our metric $TUDiv$ is similar to binomial diversity in the sense that it sets thresholds which implicitly penalize over-redundancy by taking away the reward for hitting new categories. However, the converse is not true, and the Binomial Diversity reranking method achieves poor values for both the $TIDiv$ and $TUDiv$ metrics.

Among the methods we tested, the best proxy for our $TUDiv$ metric was provided by the xQuAD approach and none of the algorithms we tested provided a good proxy for $TIDiv$. While this deficiency can be excused, as none of these algorithms take as input the various user type grouping we provide to our diversifiers, each of the other baselines also regressed or insignificantly changed the sales diversity metrics such as the aggregate diversity. This validates our hypothesis that $TIDiv$ is best thought of as a sales diversity measure and that being category-aware in the user lists is not enough for a reranking algorithm to produce diverse results for items.

\subsection{Experiments on Disjoint Categories}
\label{sec:chap3disjoint}

\begin{sidewaystable}
\begin{center}
\begin{tabular}{ |c|c|c|c|c|c|c|c|c| }
\hline
Method                        & P@10 & ERR-IA@10 & ILD@10 & BD@10 & TUDiv@10 & TIDiv@10 & A@10 & G@10 \\ \hline
TOP                           & 0.315 & 0.078 & 0.266 & 0.530 & 0.119 & 0.167 & 0.346 & 0.070\\ \hline
xQuAD ($\lambda=0.2$)         & \textbf{0.317} & 0.081 & 0.271 & 0.545 & 0.125 & 0.167 & 0.357 & 0.072\\ \hline
MMR ($\lambda=0.4$)           & 0.302 & 0.076 & \textbf{0.279} & 0.631 & 0.111 & 0.163 & 0.356 & 0.070\\ \hline
BD ($\lambda=0.8$)            & 0.285 & \textbf{0.092} & 0.268 & \textbf{0.652} & 0.113 & 0.163 & 0.372 & 0.075\\ \hline
Flow ($\beta=0.2, \mu=0.4$)   & 0.277 & 0.055 & 0.246 & 0.599 & 0.118 & 0.205 & \textbf{0.689} & 0.134\\ \hline
Flow ($\beta=0.2, \mu=0.1$)   & 0.281 & 0.056 & 0.251 & 0.601 & 0.125 & 0.201 & 0.627 & 0.120\\ \hline
Greedy ($\beta=0.2, \mu=4.0$) & 0.258 & 0.081 & 0.224 & 0.590 & 0.132 & \textbf{0.227} & 0.674 & \textbf{0.141}\\ \hline
Greedy ($\beta=8.0, \mu=0.2$) & 0.251 & 0.086 & 0.216 & 0.589 & \textbf{0.136} & 0.217 & 0.616 & 0.135\\ \hline
  \end{tabular}
\caption{hMovieLens diversifications for movie studio based categories and age group based types. The best value in each metric is bolded. \label{tbl:chapter3_disjoint_table1}}
\smallskip

\begin{tabular}{ |c|c|c|c|c|c|c|c|c| }
\hline
Method                        & P@10 & ERR-IA@10 & ILD@10 & BD@10 & TUDiv@10 & TIDiv@10 & A@10 & G@10 \\  \hline
TOP                           & 0.315 & 0.078 & 0.266 & 0.530 & 0.119 & 0.140 & 0.346 & 0.070\\ \hline
xQuAD ($\lambda=0.2$)         & \textbf{0.317} & 0.081 & 0.271 & 0.545 & 0.125 & 0.139 & 0.357 & 0.072\\ \hline
MMR ($\lambda=0.4$)           & 0.302 & 0.076 & \textbf{0.279} & 0.631 & 0.111 & 0.136 & 0.356 & 0.070\\ \hline
BD ($\lambda=0.8$)            & 0.285 & \textbf{0.092} & 0.268 & \textbf{0.652} & 0.113 & 0.135 & 0.372 & 0.075\\ \hline
Flow ($\beta=0.2, \mu=0.4$)   & 0.276 & 0.055 & 0.245 & 0.598 & 0.117 & 0.174 & \textbf{0.690} & 0.133\\ \hline
Flow ($\beta=0.2, \mu=0.1$)   & 0.282 & 0.056 & 0.251 & 0.600 & 0.122 & 0.173 & 0.618 & 0.117\\ \hline
Greedy ($\beta=0.2, \mu=4.0$) & 0.260 & 0.081 & 0.225 & 0.593 & 0.133 & \textbf{0.196} & 0.660 & \textbf{0.137}\\ \hline
Greedy ($\beta=8.0, \mu=0.2$) & 0.253 & 0.087 & 0.217 & 0.590 & \textbf{0.137} & 0.185 & 0.613 & 0.131\\ \hline
  \end{tabular}
\caption{MovieLens diversifications based on movie studio based categories and occupation based types. The best value in each metric is bolded. \label{tbl:chapter3_disjoint_table2}}
\smallskip

\begin{tabular}{ |c|c|c|c|c|c|c|c|c| }
\hline
Method                        & P@10 & ERR-IA@10 & ILD@10 & BD@10 & TUDiv@10 & TIDiv@10 & A@10 & G@10 \\  \hline
TOP                           & 0.315 & 0.078 & 0.266 & 0.530 & 0.119 & 0.184 & 0.346 & 0.070\\ \hline
xQuAD ($\lambda=0.2$)         & \textbf{0.317} & 0.081 & 0.271 & 0.545 & 0.125 & 0.185 & 0.357 & 0.072\\ \hline
MMR ($\lambda=0.4$)           & 0.302 & 0.076 & \textbf{0.279} & 0.631 & 0.111 & 0.179 & 0.356 & 0.070\\ \hline
BD ($\lambda=0.8$)            & 0.285 & \textbf{0.092} & 0.268 & \textbf{0.652} & 0.113 & 0.182 & 0.372 & 0.075\\ \hline
Flow ($\beta=0.2, \mu=0.4$)   & 0.277 & 0.055 & 0.247 & 0.599 & 0.117 & 0.220 & \textbf{0.688} & 0.136\\ \hline
Flow ($\beta=0.2, \mu=0.1$)   & 0.282 & 0.056 & 0.252 & 0.601 & 0.127 & 0.219 & 0.631 & 0.121\\ \hline
Greedy ($\beta=0.2, \mu=4.0$) & 0.260 & 0.081 & 0.226 & 0.593 & 0.133 & \textbf{0.246} & 0.669 & \textbf{0.142}\\ \hline
Greedy ($\beta=8.0, \mu=0.2$) & 0.252 & 0.087 & 0.217 & 0.592 & \textbf{0.137} & 0.236 & 0.617 & 0.137\\ \hline
\end{tabular}
\caption{MovieLens diversifications based on movie studio categories and gender based types data. The best value in each metric is bolded. \label{tbl:chapter3_disjoint_table3}}
\end{center}
\end{sidewaystable}

\begin{figure}	
\centering
 \includegraphics[width=0.95\columnwidth]{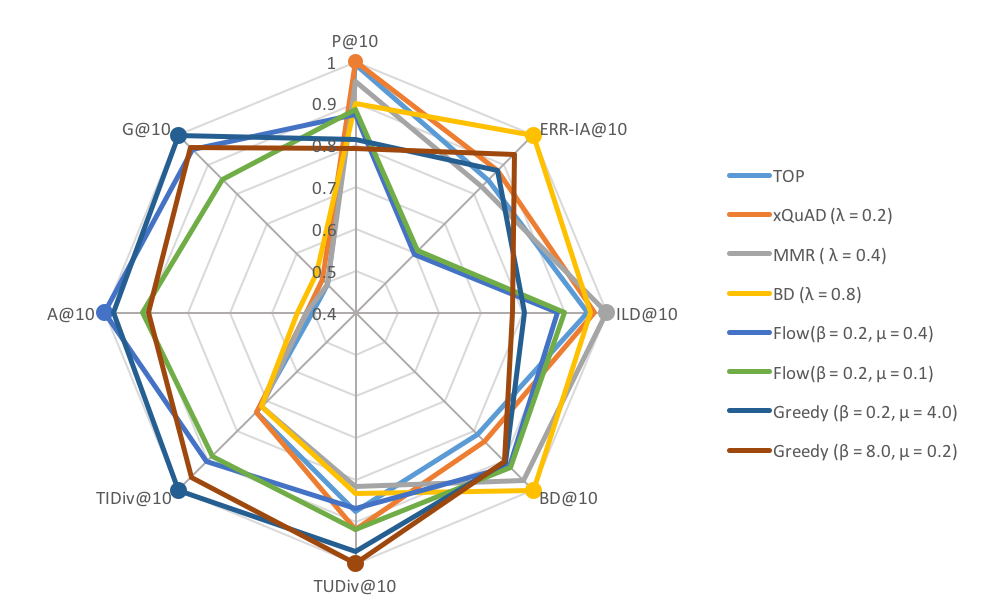}
\caption{A radial graph showing the relative performance of the reranking methods we tested for MovieLens data  with movie studio and age group based diversification. \label{fig:chapter3_disjoint_radial}}
\end{figure}

In this section we present the diversification results for disjoint item categories derived from movie studio information. This is intended to simulate the scenario where a content aggregator would like to diversify recommendations among different content providers. Since we can apply both the greedy algorithm and the flow based algorithm in this case, we report results for both. Our results for the top-10 recommendation diversification task are summarized in \ifdefined\THESIS Tables ~\ref{tbl:chapter3_overlap_table1},\ref{tbl:chapter3_overlap_table2},\ref{tbl:chapter3_overlap_table3} \else Figure ~\ref{fig:chapter3_disjoint_radial} \fi. We note that once again, every reranker optimizes its metric the best, with the exception of the xQuAD, whose objective is actually maximized by the Binomial Diversity reranking method. We also note that the precision based effectiveness of our greedy algorithm is reduced in this setting, while its effectiveness in the sales diversity metrics is amplified.

Our flow-based method and greedy algorithm show several notable differences in experimental evaluation. First, we find that the greedy algorithm actually performs better than the flow-based method in our intent-aware metrics $TUDiv$ and $TIDiv$, although our flow based methods produce more accurate recommendation lists. The solution each algorithm produces creates a different trade-off between the $TUDiv$ term of the objective, the $TIDiv$ term of the objective and the predicted relevance term of the objective.
Both optimize Binomial Diversity equally well, while the flow based method increases intra-list distance and aggregate diversity better than the greedy algorithm. The two methods' overall results are similar enough that if precision is not as big a concern as intent-aware diversification, the two algorithms can be used interchangeably.

This is a significant finding as our flow-based algorithms, while more accurate, takes more time to run to completion.
In particular, our greedy algorithms have runtime proportional to $O(|E|\log(|E|))$ where $|E|$ is the number of candidate edges, while our flow-based algorithms have complexity at least $O(|E|(|R| + |L|))$ and significantly higher overheads.
Moreover, greedy is the fastest among the methods we tested, which can be seen in Table 2.

\begin{table}
\centering
\label{tbl:chap3runtime}
\begin{tabular}{|c|c|c|c|c|c|}
\hline
Method      & Greedy & Flow & MMR  & xQuAD & BD   \\ \hline
Runtime (s) & 5.83   & 20.3 & 8.18 & 11.25 & 31.3 \\ \hline
\end{tabular}
\caption{Running time of the 5 different rerankers on the diversification task in \ifdefined\THESIS Table ~\ref{tbl:chapter3_disjoint_table1} \else Figure \ref{fig:chapter3_disjoint_radial} \fi. \label{tbl:chapter3_runtime}}
\end{table}

\section{Conclusions and Future Work}

We have presented a framework for the implementation of a diversification framework that seeks to increase the exposure of every user to predefined categories of items. 
The implementation of our framework to the case of disjoint categories of items is completely novel, and provides stronger theoretical guarantees on the quality of the solution than the implementation of our framework with overlapping categories. Extending our work to dynamic settings where the new items or the users arrive and depart over time is a rich avenue for future work.

\newpage

\bibliographystyle{siamplain}

\bibliography{references}

\end{document}